\documentclass[a4paper,fleqn]{cas-sc}

\usepackage[numbers,sort&compress]{natbib}
\usepackage{amsmath}
\usepackage{amssymb}
\usepackage{mathtools}
\usepackage{booktabs}
\usepackage{tabularx}
\usepackage{array}
\usepackage{microtype}
\usepackage{flafter}
\usepackage[section]{placeins}
\usepackage{tikz}
\usetikzlibrary{positioning,calc}
\usepackage[ruled,vlined,linesnumbered]{algorithm2e}

\newtheorem{theorem}{Theorem}
\newtheorem{lemma}[theorem]{Lemma}

\newtheorem{corollary}[theorem]{Corollary}
\newdefinition{definition}[theorem]{Definition}
\newdefinition{remark}[theorem]{Remark}
\newproof{proof}{Proof}

\newcommand{\R}{\mathbb{R}}
\newcommand{\Q}{\mathbb{Q}}
\newcommand{\Z}{\mathbb{Z}}
\newcommand{\norm}[1]{\lVert #1\rVert_p}

\tikzset{
 terminal/.style={circle,draw,fill=white,inner sep=0pt,minimum size=4.2pt},
 steiner/.style={circle,draw,fill=black,inner sep=0pt,minimum size=4.2pt},
 topoedge/.style={line width=0.6pt},
 zeroedge/.style={densely dashed,line width=0.7pt}
}

\begin{document}
\let\WriteBookmarks\relax
\def\floatpagepagefraction{1}
\def\textpagefraction{.001}

\shorttitle{Exact Algorithms for Minimum Steiner Point Trees}
\shortauthors{E. Woo and D. Shin}

\title[mode=title]{Exact Algorithms for Minimum Steiner Point Trees with Bounded Edge Length}

\author[1]{Eungyu Woo}[orcid=0009-0000-3621-5175]
\ead{ekwoo@dgist.ac.kr}

\author[1]{Donghoon Shin}[orcid=0009-0004-7088-3846]
\cormark[1]
\ead{dshin@dgist.ac.kr}

\affiliation[1]{organization={Department of EECS, DGIST}, country={South Korea}}

\cortext[1]{Corresponding author}

\begin{abstract}
Given distinct terminals $P\subset\R^2$ and $R>0$, the Steiner tree problem with minimum number of Steiner points and bounded edge length asks for a straight line tree spanning $P$, with every edge of length at most $R$, that minimizes the number of Steiner points. Length is measured in a fixed $L_p$ metric with $p\in\Q_{\ge 1}\cup\{\infty\}$. The optimum $k$ is not bounded by $n$, even in two-terminal case.

We give a deterministic exact algorithm that computes an optimal implicit representation in $n^{O(n)}$ time, independent of $k$, in the computation model of Section~\ref{subseccomputation}. The representation consists of a full Steiner topology, exact branch coordinates, and a segment count for each topology edge. Subdivision requires additional time $\Theta(n+k)$.

For each full Steiner topology, the feasible segment count vectors are the integer points of a convex projection in $O(n)$ dimensions. A continuous relaxation restricts the integer optimum to $2n-3$ consecutive values. Exact semialgebraic routines and a flatness recursion in integral lattice coordinates decide these values. Together with the parameterized bottleneck algorithm of Bandyapadhyay et al., this gives the value bound $\min\{n^{O(n)},\allowbreak k^{O(k)}n^{O(1)}\}$ for every fixed metric considered here.
\end{abstract}

\begin{highlights}
\item An exact implicit algorithm solves STP-MSPBEL in $n^{O(n)}$ time.
\item The bound is independent of the optimum number of Steiner points.
\item Each topology gives $O(n)$ integer feasibility tests in $O(n)$ dimensions.
\item Explicit expansion takes $\Theta(n+k)$ time.
\end{highlights}

\begin{keywords}
Steiner point \sep Bounded edge length \sep Exact algorithm \sep Full Steiner topology \sep Semialgebraic feasibility \sep Parameterized algorithm
\end{keywords}

\maketitle

\section{Introduction}
\label{secintro}

Lin and Xue introduced the Steiner tree problem with minimum number of Steiner points and bounded edge length (STP-MSPBEL), proved it NP-hard, and gave a minimum spanning tree heuristic with approximation factor five~\cite{LinXue1999}. Subsequent work improved approximation guarantees in Euclidean and normed planes~\cite{CalinescuWang2023,ChenEtAl2001TCS,ChengEtAl2008,CohenNutov2018,MandoiuZelikovsky2000,NutovYaroshevitch2009,ShinChoi2023}. Exact algorithms have been studied for bottleneck Steiner tree problems with a given Steiner point budget~\cite{BaeEtAl2011,BrazilEtAl2011Generalised}. Bandyapadhyay et al. gave a $k^{O(k)}n^{O(1)}$ algorithm for Euclidean Bottleneck Steiner Tree and state that it extends to the $\ell_p$ metric for $p\in\Q_{\ge 1}\cup\{\infty\}$~\cite{BandyapadhyayEtAl2024}.

The optimum number $k$ of Steiner points in STP-MSPBEL can be arbitrarily large even for two terminals, since long terminal distances force long chains of degree two subdivision points. We separate this subdivision from the branching structure. The algorithm computes a full Steiner topology, exact branch coordinates, and segment counts for each topology in $n^{O(n)}$. The count does not depend on $k$. Expanding the representation takes additional time $\Theta(n+k)$. Combining our algorithm with the solution of dual problem~\cite{BandyapadhyayEtAl2024} gives the value bound $\min\{n^{O(n)},\allowbreak k^{O(k)}n^{O(1)}\}$ for every fixed metric considered here.

The splitting transformation, full Steiner representation, beading construction, and bead count identity are due to Brazil, Ras, and Thomas~\cite{BrazilRasThomas2010}. Our contribution begins with the resulting full topology. For a fixed full topology, the feasible edge counts are the integer points of a convex projection in $O(n)$ dimensions. Its continuous optimum leaves only $2n-3$ candidate objective values. A compact lifted semialgebraic description gives exact optimization and point recovery for every body created by the recursion. We then adapt the fixed dimensional recursion of Lenstra and the width norm formulation of Dadush, Peikert, and Vempala~\cite{Lenstra1983,DadushPeikertVempala2011}. Lower dimensional recursive bodies are reduced by a thin rational ellipsoid.

Section~\ref{secprelim} gives the definitions and computation model. Section~\ref{secgeneral} proves the full topology formulation, exact algorithm, output bound, and bottleneck threshold relation. Appendix~\ref{appflatness} proves the integer feasibility theorem. Appendix~\ref{appdualebst} gives the complementary full topology bound for the dual bottleneck problem.

\section{Preliminaries}
\label{secprelim}

\subsection{Problem and terminology}
\label{subsecproblem}

Fix $p\in\Q_{\ge 1}\cup\{\infty\}$ and measure length in the $L_p$ norm. Constants hidden in asymptotic notation may depend on this fixed metric. For distinct terminals $P=\{p_1,\ldots,p_n\}\subset\R^2$ and $R\in\R_{>0}$, STP-MSPBEL asks for a straight line tree spanning $P$ with edge length at most $R$ and the minimum number $k$ of Steiner points. Vertices in $P$ are called \emph{terminals} and every other vertex is a \emph{Steiner point}.

\begin{definition}
\label{deffulltopology}
A \emph{full Steiner topology} on $n$ terminals is an unrooted abstract tree whose leaves are exactly the $n$ labeled terminals and the others are unlabeled and have degree three.
\end{definition}

This is the abstract topology of a full Steiner tree in the terminology of Brazil, Ras, and Thomas~\cite{BrazilRasThomas2010}. A full Steiner topology has $n-2$ unlabeled internal vertices and $2n-3$ edges. We write $\mathcal F_n$ for the set of all such topologies.

An \emph{implicit representation} is a triple $(H,x,q)$ formed by a full Steiner topology $H$, exact coordinates $x$ for its unlabeled vertices, and counts $q_e\in\Z_{\ge0}$. An edge with count zero has coincident endpoints and is contracted. Every positive edge is divided into $q_e$ equal segments. The degree two subdivision points are not stored.

Let $B(P)$ be axis-parallel bounding box of $P$, and let $\Delta_B$ be its diameter in the $L_p$ norm. For $y=(y_1,y_2)\in\R^2$ and $t\in\R$, let $\mathcal B_p(y,t)$ be an existential polynomial formula of constant size for $t\ge0$ and $\norm{y}\le t$. For $p=1$, we use
$$
 t\ge0,
 \qquad
 \sigma_1y_1+\sigma_2y_2\le t
 \quad\text{for all }(\sigma_1,\sigma_2)\in\{\pm1\}^2.
$$
For $p=\infty$, we use $t\ge0$ and $-t\le y_i\le t$ for $i=1,2$. For $p=a/b>1$ with coprime positive integers $a$ and $b$, one occurrence of $\mathcal B_p(y,t)$ introduces local variables $r_1,r_2,s\ge0$ and uses
$$
 t\ge0,
 \qquad
 s^b=t,
 \qquad
 r_i^{2b}=y_i^2~(i=1,2),
 \qquad
 r_1^a+r_2^a\le s^a.
$$
Thus each $\mathcal B_p$ occurrence contributes constant variables and degree depending only on the fixed metric.

\subsection{Computation model}
\label{subseccomputation}

The terminal coordinates and $R$ are represented elements of an initial ordered domain $\mathbb D_0$ in a real closed field $\mathbb K\subseteq\R$. Semialgebraic calls use arithmetic complexity over the current ordered coefficient domain $\mathbb D$, as in Basu, Pollack, and Roy~\cite[Chapter~8]{BasuPollackRoy2006}. Their cost is the number of arithmetic operations and sign tests in $\mathbb D$. Algebraic outputs are stored by real univariate representations. When such an output is used as a coefficient in a later call, its representation is adjoined to the coefficient domain and operations in the extension are charged by the same convention. Coefficient bit lengths and coefficient growth are not charged.

Represented scalar field operations, sign tests, floor, ceiling, extended-gcd on integers, and the fixed-degree algebraic operations required by Euclidean and fixed $L_p$ norms take unit cost. Vector and matrix computations are charged by their scalar operations. A primitive vector $a\in\Z^r$ is completed to a unimodular coordinate system satisfying $a^TW=e_r^T$ by $O(r)$ such transformations, stored as their product.

\subsection{Semialgebraic and integer feasibility tools}
\label{subsecpublishedtools}

We use exact emptiness testing, point recovery, and exact linear and quadratic optimization over compact convex sets. The same operations remain available, with polynomial overhead in the dimension, under affine sections, injective affine substitutions, and homotheties. Represented outer radius data are maintained under the same transformations.

\begin{theorem}
\label{thmsemialgopt}
Let $\widehat S\subseteq\mathbb K^N$ be a compact basic closed semialgebraic set defined by $s$ polynomial equations and nonstrict inequalities with coefficients in $\mathbb D$ and degree at most $\delta$. Let $S$ be its projection onto prescribed coordinates. Emptiness of $S$ can be decided in $(s\delta)^{O(N)}$ operations in $\mathbb D$. When $S$ is nonempty, the same bound returns a point of $S$ and a lift in $\widehat S$. If $f$ is a polynomial of degree at most $\delta$ in the projected coordinates, its minimum on $S$ and one lifted minimizer are computable within the same bound.
\end{theorem}

\begin{proof}
Apply the critical point sampling and one-block elimination algorithms of Basu, Pollack, and Roy~\cite[Theorems~13.22 and~14.16]{BasuPollackRoy2006} to $\widehat S$. They decide emptiness and return a point of $\widehat S$ in $(s\delta)^{O(N)}$ operations. Projection gives the stated point of $S$ and its lift. For optimization, write $\pi$ for the prescribed projection and apply the same algorithms to the compact graph
$$
 \{(u,t)\in\widehat S\times\mathbb K\mid t=f(\pi(u))\}.
$$
One-block elimination and univariate root isolation give the least projected value of $t$, and critical point sampling in that fiber gives a lifted minimizer. Adding one variable and one equation preserves the stated bound.
\end{proof}

\begin{theorem}[Theorem~1.6 in~\cite{DadushVempala2012}]
\label{thmgeneralnormsvp}
Given a basis for a lattice $L$ and a well-centered norm $\lVert\cdot\rVert_K$ specified by a convex body $K$, both in $\R^r$, the shortest vector in $L$ under $\lVert\cdot\rVert_K$ can be found deterministically in $2^{O(r)}$ time and space.
\end{theorem}

The cited bound counts membership queries and arithmetic operations. A centrally symmetric convex body is well-centered. In our application, a unit ball query uses two exact linear optimizations, and Appendix~\ref{appflatness} supplies a Euclidean inner and outer bound whose ratio depends only on the dimension.

The next theorem gives the convex integer feasibility result used below. Its recursion follows Lenstra~\cite{Lenstra1983}. The width norm and wide body branch follow the proof of Theorem~4.7 in~\cite{DadushPeikertVempala2011}, while the exact shortest vector step uses Theorem~\ref{thmgeneralnormsvp}. The thin rational reduction and the operation count in the present model are proved in Appendix~\ref{appflatness}. For $r\ge1$, $x\in\R^r$, and $\rho\ge0$, $B_2(x,\rho)=\{y\in\R^r\mid \|y-x\|_2\le\rho\}$, $B_2^r=B_2(0,1).$

\begin{theorem}
\label{thmconvexintegerfeas}
Let $K\subseteq\R^r$ be compact and convex with represented outer radius data
$$
 K\subseteq a_0+R_{\rm out}B_2^r.
$$
Let $A$ bound the operations for $K$ and for every body obtained recursively by affine sections, injective affine substitutions, and homotheties. A deterministic algorithm decides whether $K\cap\Z^r$ is empty and returns a point when one exists in $r^{O(r)}A$. In particular, the time is $n^{O(n)}$ when $r=O(n)$ and $A=n^{O(n)}$.
\end{theorem}

\section{STP-MSPBEL}
\label{secgeneral}

We first represent every optimum on a full Steiner topology, then optimize and enumerate these topologies.

\subsection{Full Steiner representation}
\label{subsecfullrepresentation}

The pruning, straightening, and splitting reductions below are the $R$-scaled form of the full tree and beading framework in~\cite[Section~2]{BrazilRasThomas2010}. We include the short proofs because zero length edges are used later.

\begin{lemma}
\label{lemnoleaf}
An optimal tree has no Steiner leaf.
\end{lemma}

\begin{proof}
A Steiner leaf is a Steiner point of degree one. Deleting it together with its incident edge preserves terminal connectivity and decreases the number of Steiner points.
\end{proof}

\begin{lemma}
\label{lemstraight}
A path with $q$ degree two Steiner vertices can be replaced by $q$ equally spaced points on the segment between its endpoints.
\end{lemma}

\begin{proof}
The triangle inequality bounds the distance between the endpoints by $(q+1)R$. Equal subdivision preserves the edge bound.
\end{proof}

Suppress every degree two Steiner point of an optimal tree while retaining all terminals. By Lemma~\ref{lemnoleaf}, the resulting \emph{reduced topology} has no unlabeled vertex of degree one or two. The splitting transformation replaces internal terminals and vertices of degree greater than three by cubic trees embedded at the original position, with zero length internal edges~\cite[Section~2]{BrazilRasThomas2010}.

\begin{lemma}
\label{lemsplitting}
Every reduced topology has a full Steiner representation obtained by splitting vertices and adding zero length edges. The representation has $n-2$ auxiliary vertices.
\end{lemma}

\begin{proof}
Replace a terminal of degree $d$ by a cubic tree with one labeled terminal leaf and $d$ attachment edges for the original incidences. This contributes $d-1$ auxiliary vertices. Replace a nonterminal vertex of degree $d\ge3$ by a cubic tree with $d$ attachment edges, which contributes $d-2$ auxiliary vertices. Place all auxiliary vertices at the coordinate of the vertex being split. Its internal edges, including the terminal leaf edge in the first construction, have length zero.

Let $B$ be the set of nonterminal vertices of the reduced topology. The number of auxiliary vertices is
$$
 \sum_{p_i\in P}(\deg(p_i)-1)+\sum_{v\in B}(\deg(v)-2).
$$
For a tree,
$$
 \sum_{p_i\in P}(\deg(p_i)-2)+\sum_{v\in B}(\deg(v)-2)=-2.
$$
Adding $n$ gives $n-2$.
\end{proof}

\begin{remark}
\label{remdegreetwoterminal}
Splitting preserves the given position of a degree two terminal. If $p_2$ is incident to $p_1p_2$ and $p_2s$, the split introduces an auxiliary vertex $s_2$ at $p_2$ and zero-length edge $p_2s_2$. Contracting this edge restores the internal terminal.
\end{remark}

\begin{figure}[pos=htbp]
\centering
\begin{tikzpicture}[x=1cm,y=1cm,font=\footnotesize]
 \node[steiner,label=below:$v$] (va) at (0,0) {};
 \node[terminal,label=left:$p_1$] (a1) at (-1.15,0.72) {};
 \node[terminal,label=left:$p_2$] (a2) at (-1.15,-0.72) {};
 \node[terminal,label=right:$p_3$] (a3) at (1.15,0.72) {};
 \node[terminal,label=right:$p_4$] (a4) at (1.15,-0.72) {};
 \path[topoedge] (va) edge (a1) edge (a2) edge (a3) edge (a4);
 \node at (0,-1.28) {(a) reduced junction};

 \node[steiner,label=below:$s_1$] (s1) at (3.55,0) {};
 \node[steiner,label=below:$s_2$] (s2) at (4.45,0) {};
 \node[terminal,label=left:$p_1$] (b1) at (2.55,0.72) {};
 \node[terminal,label=left:$p_2$] (b2) at (2.55,-0.72) {};
 \node[terminal,label=right:$p_3$] (b3) at (5.45,0.72) {};
 \node[terminal,label=right:$p_4$] (b4) at (5.45,-0.72) {};
 \path[topoedge] (s1) edge (b1) edge (b2);
 \path[topoedge] (s2) edge (b3) edge (b4);
 \path[zeroedge] (s1) edge node[above]{$q_e=0$} (s2);
 \node at (4,-1.28) {(b) full representation};

 \node[terminal,label=left:$u$] (u) at (7.15,0) {};
 \node[terminal,label=right:$v$] (w) at (10.15,0) {};
 \path[topoedge] (u) edge (w);
 \foreach \x in {7.90,8.65,9.40}{\node[steiner] at (\x,0) {};}
 \node at (8.65,0.48) {$q_e=4$};
 \node at (8.65,-1.28) {(c) edge expansion};
\end{tikzpicture}
\caption{The two operations behind a full Steiner representation. A degree four junction is split into degree three vertices joined by a zero length edge. Beading divides an edge into $q_e$ bounded length segments.}
\label{figfullrepresentation}
\end{figure}
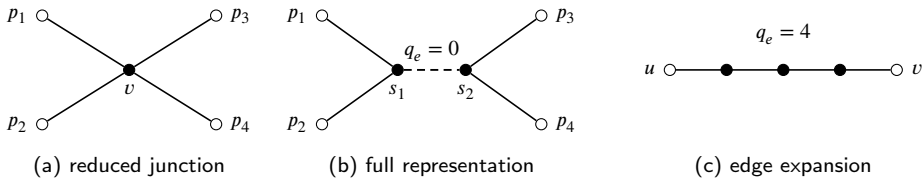
\FloatBarrier

\begin{lemma}
\label{lemfullcount}
The number of full Steiner topologies on $n\ge3$ labeled terminals is $F_n=(2n-5)!!=(2n-4)!/(2^{n-2}(n-2)!)$.
\end{lemma}

\begin{proof}
There is one full Steiner topology for $n=3$. Removing a prescribed terminal and suppressing its degree two neighbor is inverse to subdividing one of the $2n-5$ edges of a topology on $n-1$ terminals and attaching the prescribed terminal. Hence $F_n=(2n-5)F_{n-1}$ with $F_3=1$.
\end{proof}

\subsection{Bead formulation}
\label{subsecbeads}

Fix a full Steiner topology $H=(V,E)$, and let $S(H)$ be its set of unlabeled degree three vertices. Each $s\in S(H)$ has a branch coordinate $x_s\in\R^2$. For $v\in V$, write
$$
X_v(x)=
\begin{cases}
 p_i, & v\text{ is the terminal leaf labeled }i,\\
 x_v, & v\in S(H).
\end{cases}
$$
For an edge $e=uv$, define its segment count at $x$ by
\begin{equation}
\label{eqbeadnumber}
 q_e(x)=
 \left\lceil
 \frac{\norm{X_u(x)-X_v(x)}}{R}
 \right\rceil.
\end{equation}
Thus $q_e(x)=0$ exactly when $e$ has zero length. Otherwise, beading $e$ produces $q_e(x)$ equal segments and $q_e(x)-1$ degree two Steiner points. Brazil, Ras, and Thomas define beading and give the formula $1-n+\sum_i\lceil |e_i|\rceil$ after converting the tree into full form~\cite[Section~2]{BrazilRasThomas2010}. For edge length bound $R$, the objective for a fixed topology is
\begin{equation}
\label{eqbranchobjective}
 \mu_H=
 \min_{x\in B(P)^{S(H)}}
 \sum_{e=uv\in E(H)}
 \left\lceil
 \frac{\norm{X_u(x)-X_v(x)}}{R}
 \right\rceil.
\end{equation}
Coordinatewise projection onto $B(P)$ fixes every terminal and is nonexpansive in the $L_p$ norm, so the restriction to $B(P)^{S(H)}$ is without loss of generality. The minimum in \eqref{eqbranchobjective} is attained because the objective is lower semicontinuous on this compact domain.

\begin{lemma}
\label{lembeadcount}
Let $k^*$ be the optimum number of Steiner points in the STP-MSPBEL instance. Then $k^*=\min_{H\in\mathcal F_n}(\mu_H-n+1)$, where $\mathcal F_n$ is the set of full Steiner topologies on the $n$ labeled terminals.
\end{lemma}

\begin{proof}
Fix $H$ and branch coordinates $x$, set $q_e=q_e(x)$, and let $z=|\{e\mid q_e(x)=0\}|$. Contracting them leaves $2n-2-z$ vertices. Distinct terminals lie in distinct components formed by these contractions, so $n-2-z$ remaining vertices are nonterminal. Beading the edges adds $\sum_{q_e>0}(q_e-1)=\sum_{e\in E(H)}q_e-(2n-3-z)$ degree two Steiner points. Before coincident vertices are identified, the total number of Steiner vertices is therefore
$$
 (n-2-z)+\sum_{e\in E(H)}q_e-(2n-3-z)
 =1-n+\sum_{e\in E(H)}q_e.
$$
All expanded edges have length at most $R$. Identifying coincident vertices and deleting one edge from each resulting cycle preserves feasibility and cannot increase this number. Hence
$$
 k^*\le \min_{H\in\mathcal F_n}(\mu_H-n+1).
$$

Conversely, identifying coincident vertices of an optimum tree and deleting one edge from each resulting cycle preserve feasibility and cannot increase the number of Steiner points. A strict decrease would contradict optimality, so the resulting tree still has $k^*$ Steiner points and distinct vertex positions. By Lemma~\ref{lemnoleaf}, it has no Steiner leaf. Apply Lemma~\ref{lemstraight} to every maximal degree two Steiner path, suppress its internal vertices, and split the reduced topology by Lemma~\ref{lemsplitting}. Let $q_e$ be the number of segments of length at most $R$ represented by a positive full topology edge and set $q_e=0$ on split edges. The same count gives
$$
 k^*=1-n+\sum_{e\in E(H)}q_e.
$$
Projecting the auxiliary positions onto $B(P)$ cannot increase any edge length and preserves split edges. Thus the vector $q$ remains feasible and
$$
 \mu_H-n+1\le1-n+\sum_{e\in E(H)}q_e=k^*.
$$
\end{proof}

\subsection{Optimization for a fixed topology}
\label{subsecfixed}

We optimize \eqref{eqbranchobjective} in the space of segment count vectors, without enumerating those vectors. For a fixed full Steiner topology $H=(V,E)$, define the continuous segment count set
\begin{equation}
\label{eqcountbody}
C_H=
\left\{q\in\R_{\ge0}^{E} \,\middle|\,
\begin{array}{l}
\exists x\in B(P)^{S(H)},~\mathcal B_p(X_u(x)-X_v(x),Rq_e)\text{ for all }e=uv\in E
\end{array}
\right\}.
\end{equation}
The set $C_H$ is convex. If $q^1,q^2\in C_H$ have witnesses $x^1,x^2$ and $t\in[0,1]$, then $x=tx^1+(1-t)x^2\in B(P)^{S(H)}$ and
$$
 \norm{X_u(x)-X_v(x)}
 \le t\norm{X_u(x^1)-X_v(x^1)}
 +(1-t)\norm{X_u(x^2)-X_v(x^2)}
 \le R\bigl(tq^1_e+(1-t)q^2_e\bigr)
$$
for every $e=uv$. Hence $tq^1+(1-t)q^2\in C_H$. Closedness follows by taking a convergent subsequence of witnesses in the compact set $B(P)^{S(H)}$.

\begin{lemma}
\label{lemcountbodyequivalence}
For every full topology $H$, $\mu_H=\min\{\mathbf 1^Tq\mid q\in C_H\cap\Z_{\ge0}^{E}\}$.
\end{lemma}

\begin{proof}
For branch coordinates $x$, the induced vector $q_e=\left\lceil {\norm{X_u(x)-X_v(x)}}/{R}\right\rceil$
belongs to $C_H\cap\Z_{\ge0}^{E}$ and has objective value equal to that of $x$ in \eqref{eqbranchobjective}. Thus the expression on the right is at most $\mu_H$.

Conversely, let $q\in C_H\cap\Z_{\ge0}^{E}$ and let $x$ witness $q\in C_H$. Then $\left\lceil{\norm{X_u(x)-X_v(x)}}/{R}\right\rceil\le q_e~(e=uv\in E),$
so the objective at $x$ is at most $\mathbf 1^Tq$. Taking minima gives the reverse inequality.
\end{proof}

The set $C_H$ is upward closed under the coordinatewise order. Define its continuous optimum by
\begin{equation}
\label{eqcontinuouscountrelaxation}
 \lambda_H=
 \min\{\mathbf 1^Tq\mid q\in C_H\}.
\end{equation}
Since all branch coordinates are restricted to $B(P)$, every edge length is at most $\Delta_B$. The same value is attained by the compact semialgebraic program
\begin{equation}
\label{eqcompactcontinuouscountrelaxation}
\begin{aligned}
 \lambda_H=\min \quad & \mathbf 1^Tq \\
 \text{s.t.}\quad
 &x_s\in B(P) \qquad(s\in S(H)),\\
 &\mathcal B_p(X_u(x)-X_v(x),Rq_e) \qquad(e=uv\in E),\\
 &0\le q_e\le \Delta_B/R \qquad(e\in E).
\end{aligned}
\end{equation}
For fixed $x$, the componentwise minimum is $q_e=\norm{X_u(x)-X_v(x)}/R\le\Delta_B/R$. Thus \eqref{eqcompactcontinuouscountrelaxation} has value $\lambda_H$. Its expanded formulation has $O(n)$ variables and constraints of constant degree, so Theorem~\ref{thmsemialgopt} computes $\lambda_H$ and an exact optimizer in $n^{O(n)}$.

Brazil, Ras, and Thomas obtain the same additive quantity $2n-4$ in a global comparison between a beaded Steiner minimal tree and an optimum minimum Steiner point tree~\cite[Proposition~3]{BrazilRasThomas2010}. The next lemma is the fixed topology convex rounding form used by our algorithm.

\begin{lemma}
\label{lemshortobjectivewindow}
For every full topology $H$, the integer optimum $\mu_H$ satisfies
$$
 \lceil\lambda_H\rceil
 \le
 \mu_H
 \le
 \lceil\lambda_H\rceil+|E(H)|-1
 =
 \lceil\lambda_H\rceil+2n-4.
$$
Thus at most $|E(H)|=2n-3$ objective values have to be tested.
\end{lemma}

\begin{proof}
Since $C_H\cap\Z_{\ge0}^{E}\subseteq C_H$ and $\mu_H$ is integral, $\mu_H\ge\lceil\lambda_H\rceil$. Let $\bar q$ be a continuous optimum. Upward closure gives $\lceil\bar q\rceil\in C_H\cap\Z_{\ge0}^{E}$, and therefore
$$
 \mu_H\le\mathbf 1^T\lceil\bar q\rceil
 <\mathbf 1^T\bar q+|E(H)|
 =\lambda_H+|E(H)|.
$$
The upper bound follows by integrality of $\mu_H$.
\end{proof}

Define the threshold body
$$
 K_{H,M}=C_H\cap\{q\in\R_{\ge0}^{E}\mid\mathbf 1^Tq\le M+\frac12\}.
$$
For integer $q$, the last inequality is equivalent to
$\mathbf 1^Tq\le M$. Hence $\mu_H\le M$ if and only if
$K_{H,M}\cap\Z^E\ne\emptyset$. Write $d=|E|=2n-3$.

For $M\ge\lceil\lambda_H\rceil$, let $\bar q$ be a continuous optimum
and set
$$
 \gamma=M+\frac12-\lambda_H\ge\frac12,\qquad q^0=\bar q+\frac{\gamma}{4d}\mathbf 1,\qquad\rho=\frac{\gamma}{4d}.
$$
For every $q\in B_2(q^0,\rho)$ and every coordinate $e$,
$$q_e\ge q^0_e-\rho=\bar q_e.$$
Moreover,
$$
 \mathbf 1^Tq
 \le
 \lambda_H+\frac{\gamma}{4}+\sqrt d\,\rho
 =
 \lambda_H+\frac{\gamma}{4}+\frac{\gamma}{4\sqrt d}
 \le
 \lambda_H+\frac{\gamma}{2}
 <
 M+\frac12.
$$
Upward closure therefore gives $B_2(q^0,\rho)\subseteq K_{H,M}.$

If $q\in K_{H,M}$, then $\|q\|_2\le \mathbf 1^Tq\le M+\frac12$, while $
 \|q^0\|_2
 \le \mathbf 1^Tq^0
 =\lambda_H+\frac{\gamma}{4}
 \le M+\frac12
$.
Hence
$$
 \|q-q^0\|_2
 \le \|q\|_2+\|q^0\|_2
 \le 2(M+\frac12)
 =2M+1.
$$
$$K_{H,M}\subseteq B_2(q^0,2M+1).$$
The body is full dimensional and has represented outer radius data.

\begin{lemma}
\label{lemkoracles}
For a fixed full topology $H$ and an integer $M\ge\lceil\lambda_H\rceil$, the operations required by Theorem~\ref{thmconvexintegerfeas} cost $n^{O(n)}$ on $K_{H,M}$ and on every body produced by the recursion. A returned lattice point can be mapped to $K_{H,M}\cap\Z^E$, and compatible branch coordinates can be recovered within the same bound.
\end{lemma}

\begin{proof}
Let $L_{H,M}(q,x,\eta)$ consist of the box constraints on the branch coordinates $x$, the local variables $\eta$ in the formulas $\mathcal B_p$, the edge constraints, $q\ge0$, and $\mathbf 1^Tq\le M+\frac12$. These constraints define a compact lift with $O(n)$ variables and constraints and degree depending only on the fixed metric.

For each recursive body $D\subseteq\R^r$, store an injective affine substitution $\sigma_D\colon\R^r\to\R^E$ and a reconstruction map $\rho_D\colon\R^r\to\R^E$ such that
$$
 D=\left\{z\in\R^r\,\middle|\,\exists x,\eta\quad
 L_{H,M}(\sigma_D(z),x,\eta)\right\},
 \qquad
 \rho_D(D\cap\Z^r)\subseteq K_{H,M}\cap\Z^E.
$$
Both maps are the identity at the root. For an integral section, compose both maps with its unimodular affine parametrization, which maps $\Z^{r-1}$ bijectively onto the integer points of the section. For a homothetic copy $D_0=h(D)\subseteq D$, use $\sigma_D\circ h^{-1}$ in the lifted formula and keep $\rho_D$ as the return map.

Since $\sigma_D$ is injective, compactness of the original lift implies compactness of every derived lift. Theorem~\ref{thmsemialgopt} therefore gives emptiness, point recovery, and exact linear and quadratic optimization in $n^{O(n)}$ time.

Outer radius data are preserved as well. Homotheties scale the radius. For a pullback
$$
 D'=\{t\in\R^k\mid At+b\in D\},
$$
where $A$ has full column rank, recover $t_0\in D'$. If $D$ lies in a ball of radius $R_0$, then
$$
 \|t-t_0\|_2
 \le 2R_0\|(A^TA)^{-1}A^T\|_F
 \qquad(t\in D').
$$
After $q\in K_{H,M}\cap\Z^E$ has been returned, fix $q$ in $L_{H,M}$ and apply Theorem~\ref{thmsemialgopt} to recover compatible branch coordinates.
\end{proof}

\begin{lemma}
\label{lemlenstrathreshold}
For a fixed topology $H$ and an integer $M\ge\lceil\lambda_H\rceil$, the predicate $\mu_H\le M$
is decided in $n^{O(n)}$. If it is true, the algorithm
returns a feasible integer count vector together with witness branch
coordinates.
\end{lemma}
\begin{proof}
Using the outer radius data above and Lemma~\ref{lemkoracles}, Theorem~\ref{thmconvexintegerfeas} decides whether $K_{H,M}\cap\Z^E$ is empty in $n^{O(n)}$.

If an integer point $q$ is returned, recover witness coordinates $x$ by Lemma~\ref{lemkoracles} and define
$$
 q^*_e=
 \left\lceil\frac{\norm{X_u(x)-X_v(x)}}{R}\right\rceil.
$$
The vector $q^*$ belongs to $C_H\cap\Z_{\ge0}^{E}$ and satisfies $q^*\le q$ coordinatewise. Since $q$ is integral and $\mathbf 1^Tq\le M+\frac12$, we have $\mathbf 1^Tq^*\le\mathbf 1^Tq\le M$. Return $(q^*,x)$.
\end{proof}

\begin{remark}
\label{remcontinuousnotexact}
The relaxation \eqref{eqcontinuouscountrelaxation} is used only for the objective window in Lemma~\ref{lemshortobjectivewindow}. Convex integer feasibility over $K_{H,M}$ finds the integer optimum. No branch coordinates or minimizer of total length are rounded.
\end{remark}

\subsection{Algorithm and complexity}
\label{subsecalgorithm}

The resulting implicit algorithm is given below.

\begin{algorithm}[H]
\caption{Exact STP-MSPBEL}\label{algbranchfulltopology}
\KwData{terminals $P=\{p_1,\ldots,p_n\}$ and edge bound $R$}
\KwResult{an optimal implicit representation}
\If{$n=1$}{
  \KwRet{the empty tree}\;
}
\If{$n=2$}{
  $q\leftarrow\left\lceil\norm{p_1-p_2}/R\right\rceil$\;
  \KwRet{the single topology edge with segment count $q$}\;
}
$\textsc{best}\leftarrow\infty$ and $\textsc{sol}\leftarrow\bot$\;
\ForEach{full Steiner topology $H$ on the labeled terminals}{
  Compute $\lambda_H$ and a represented continuous optimizer $\bar q_H$ with witness $\bar x_H$ by \eqref{eqcompactcontinuouscountrelaxation}\;
  \For{$M=\lceil\lambda_H\rceil$ \KwTo $\lceil\lambda_H\rceil+|E(H)|-1$}{
    Test $K_{H,M}\cap\Z^{E(H)}$ using Lemma~\ref{lemlenstrathreshold}\;
    \If{a pair $(q^*,x_H)$ is returned}{
      $v\leftarrow\mathbf 1^Tq^*-n+1$\;
      \If{$v<\textsc{best}$}{
        $\textsc{best}\leftarrow v$ and $\textsc{sol}\leftarrow(H,x_H,q^*)$\;
      }
      \textbf{break}\;
    }
  }
}
\KwRet{$\textsc{sol}$}
\end{algorithm}

The preceding construction gives an inner ball and an outer-radius bound for each initial threshold body. The normalization in Appendix~\ref{appflatness} depends only on the dimension.

\begin{theorem}
\label{thmcorrectness}
The algorithm returns an optimal implicit representation.
\end{theorem}

\begin{proof}
For $n=1$, the empty tree is optimal. For $n=2$, every feasible tree contains a path from $p_1$ to $p_2$, and equal subdivision of the segment $p_1p_2$ is optimal by the triangle inequality.

Assume $n\ge3$ and fix a full Steiner topology $H$. By Lemma~\ref{lemshortobjectivewindow}, the tested interval contains $\mu_H$. Lemma~\ref{lemlenstrathreshold} decides every threshold exactly, so the first feasible threshold is $M_H=\mu_H$. At this threshold the returned vector $q^*$ is integer feasible for $C_H$ and satisfies $\mathbf 1^Tq^*\le M_H$. Lemma~\ref{lemcountbodyequivalence} gives
$$
 M_H=\mu_H\le\mathbf 1^Tq^*\le M_H,
$$
so $\mathbf 1^Tq^*=\mu_H$. The algorithm selects a topology minimizing $\mu_H-n+1$, which equals the optimum STP-MSPBEL value by Lemma~\ref{lembeadcount}. The stored topology, branch coordinates, and segment counts are therefore an optimal implicit representation.
\end{proof}

\begin{theorem}
\label{thmdetcomplexity}
An optimal implicit representation for $n\ge3$ is computed deterministically in $n^{O(n)}$. The operation count is independent of the optimum number of Steiner points.
\end{theorem}

\begin{proof}
Lemma~\ref{lemfullcount} gives $(2n-5)!!=n^{\Theta(n)}$ full Steiner topologies. For each topology, the continuous relaxation takes time $n^{O(n)}$, and Lemma~\ref{lemshortobjectivewindow} requires at most $2n-3$ threshold tests. Each test and witness recovery takes time $n^{O(n)}$ by Lemma~\ref{lemlenstrathreshold}. The time per topology is therefore $n^{O(n)}$, and multiplication by $(2n-5)!!$ preserves this bound.
\end{proof}

For $n=2$, the optimum is $\max\{0,\lceil\norm{p_1-p_2}/R\rceil-1\}$ and is computed in constant time. For $n=3$, the unique full Steiner topology has one Steiner vertex, and its objective is
$$
 \min_{x\in B(P)}\sum_{i=1}^3
 \left\lceil\frac{\norm{x-p_i}}{R}\right\rceil-2.
$$
This is the branching problem for three terminals studied by Shin and Choi~\cite{ShinChoi2023}.

\subsection{Reporting}
\label{subsecreporting}

\begin{theorem}
\label{thmreporting}
Given an optimal implicit representation with $k$ Steiner points, an explicit optimum is reported in $\Theta(n+k)$.
\end{theorem}

\begin{proof}
Let $(H,x,q)$ be an optimal implicit representation and define
$$
 q'_e=\left\lceil\frac{\norm{X_u(x)-X_v(x)}}{R}\right\rceil.
$$
Then $q'\le q$. A strict decrease in the total count would give, by Lemma~\ref{lembeadcount}, a feasible tree with fewer than $k$ Steiner points. Hence $\mathbf 1^Tq'=\mathbf 1^Tq$.

First contract every topology edge with count zero. Bead each remaining edge into $q'_e$ equal segments. The counting identity in the proof of Lemma~\ref{lembeadcount} gives exactly
$$
 k=1-n+\sum_e q'_e
$$
Steiner vertices at this stage. Suppose that two distinct remaining vertices have the same position. The terminals are distinct, so at least one of them is a Steiner vertex. Identifying the pair, removing loops, and deleting one edge from each resulting cycle would preserve feasibility and reduce the number of Steiner points. This contradicts optimality. No further identification is possible. The expansion is therefore optimal, and listing its vertices and edges takes time $\Theta(n+k)$.
\end{proof}

\begin{corollary}
\label{corexplicitcomplexity}
For $n\ge2$, an explicit optimum with $k$ Steiner points is computed deterministically in $n^{O(n)}+O(n+k)$.
\end{corollary}

The corollary follows from the closed form for $n=2$ and Theorems~\ref{thmdetcomplexity} and~\ref{thmreporting}.

\subsection{Bottleneck equivalence and a hybrid bound}
\label{subsechybrid}

For $j\in\Z_{\ge0}$, let
$$
 \beta_j(P)=\min_T\max_{e\in E(T)}\operatorname{len}(e),
$$
where $T$ ranges over geometric trees spanning $P$ and containing at most $j$ Steiner points. The minimum is attained after projecting Steiner coordinates into $B(P)$ and minimizing over the finitely many abstract topologies with at most $j$ Steiner vertices.

\begin{lemma}
\label{lembottleneckthreshold}
For every integer $j\ge0$, the STP-MSPBEL instance has a feasible tree with at most $j$ Steiner points if and only if $\beta_j(P)\le R$. Consequently, if $k$ is the optimum number of Steiner points, then $k=\min\{j\in\Z_{\ge0}\mid\beta_j(P)\le R\}$.
\end{lemma}

\begin{proof}
A tree with at most $j$ Steiner points is feasible for STP-MSPBEL if and only if its bottleneck is at most $R$. Minimizing the feasible budget gives the formula for $k$.
\end{proof}

\begin{theorem}[Theorem~1 in~\cite{BandyapadhyayEtAl2024}]
\label{thmbandyapadhyay}
Euclidean Bottleneck Steiner Tree with $k$ Steiner points can be solved in $k^{O(k)}\cdot n^{O(1)}$ time. In particular, the problem is fixed-parameter tractable.
\end{theorem}

The problem definition in~\cite{BandyapadhyayEtAl2024} uses a budget of at most $k$ Steiner points. The sentence following Theorem~1 extends the algorithm to the $\ell_p$ metric for $p\in\Q_{\ge 1}\cup\{\infty\}$. Hence $\beta_j(P)$ and an optimal bottleneck Steiner tree with at most $j$ Steiner points are computable in $j^{O(j)}n^{O(1)}$ time for every fixed metric considered here.

\begin{theorem}
\label{thmhybrid}
Let $k$ be the optimum STP-MSPBEL value. For every fixed
$p\in\Q_{\ge 1}\cup\{\infty\}$, the value $k$ is computed
deterministically in
$$
 \min\{n^{O(n)},\ k^{O(k)}n^{O(1)}\}.
$$
An optimal implicit representation is computed in $n^{O(n)}$ time, and an
explicit optimum is computed in
$$
 \min\{n^{O(n)}+O(n+k),\ k^{O(k)}n^{O(1)}\}.
$$
For $k=0$, the parameterized term is interpreted as polynomial in $n$.
\end{theorem}

\begin{proof}
First test connectivity of the graph on $P$ whose edges have length at most $R$. If it is connected, then $k=0$ and any spanning tree of this graph is an explicit optimum. Otherwise, apply Theorem~\ref{thmbandyapadhyay} for budgets $j=1,2,\ldots$ and stop at the first $j$ for which $\beta_j(P)\le R$. Lemma~\ref{lembottleneckthreshold} shows that this first budget is exactly $k$, and the last bottleneck call returns an explicit STP-MSPBEL optimum. The total running time is
$$
 \sum_{j=1}^{k} j^{O(j)}n^{O(1)}=k^{O(k)}n^{O(1)}.
$$

Theorem~\ref{thmdetcomplexity} computes $k$ together with an optimal implicit representation in $n^{O(n)}$ time, and Theorem~\ref{thmreporting} expands that representation in additional time $O(n+k)$.
\end{proof}

\section{Conclusion}
\label{secconclusion}

The splitting and beading reduction gives every optimum a full Steiner representation. On each full topology, the convex segment count formulation and its continuous relaxation leave $2n-3$ candidate integer objective values. Compact lifted semialgebraic descriptions support the optimization and point recovery used in the recursion. The flatness argument in Appendix~\ref{appflatness} then yields an optimal implicit representation in $n^{O(n)}$ time. This count is independent of the optimum number of Steiner points. Explicit expansion takes time $\Theta(n+k)$.

The bottleneck threshold relation and the algorithm of Bandyapadhyay et al. give the complementary value bound $\min\{n^{O(n)},\allowbreak k^{O(k)}n^{O(1)}\}$. Appendix~\ref{appdualebst} records the full topology algorithm for the dual bottleneck value with running time $(n+j)^{O(n)}$.

\appendix
\section{Convex Integer Feasibility}
\label{appflatness}

This appendix proves Theorem~\ref{thmconvexintegerfeas}. The recursion follows Lenstra~\cite{Lenstra1983}. The reduction of lattice width to shortest vector computation and the wide case are taken from the proof of Theorem~4.7 in~\cite{DadushPeikertVempala2011}.

For a full dimensional compact convex set $C\subseteq\R^r$, let
$$
 w_{\Z}(C)=
 \min_{a\in\Z^r\setminus\{0\}}
 \left(\max_{x\in C}a^Tx-\min_{x\in C}a^Tx\right).
$$
By~\cite{BanaszczykEtAl1999,Rudelson2000}, there are constants $C_{\rm f},c_{\rm f}>0$ such that
$$
 \operatorname{int}(C)\cap\Z^r=\emptyset
 \quad\Longrightarrow\quad
 w_{\Z}(C)\le C_{\rm f}r^{4/3}\log^{c_{\rm f}}(2r).
$$
Fix integers $C',m\ge1$ with $C_{\rm f}r^{4/3}\log^{c_{\rm f}}(2r)\le C'r^m$ for all $r\ge1$, and put $\phi(r)=C'r^m$.

\subsection{Dimension reduction and minimum width}

The following lemma is the rational form of the small volume case in the proof of Theorem~4.7 in~\cite{DadushPeikertVempala2011}.

\begin{lemma}
\label{lemthindimensionreduction}
Let $C\subseteq\R^r$ be nonempty, compact, and convex. Suppose that point recovery and exact linear and quadratic optimization over $C$ cost $A$, and that
$$
 C\subseteq a_0+R_{\rm out}B_2^r.
$$
In $r^{O(r)}A$ time, one can either certify that $C$ is full dimensional, certify that $C\cap\Z^r=\emptyset$, or compute a primitive $a\in\Z^r$ and $s\in\Z$ such that
$$
 C\cap\Z^r=C\cap\{z\in\R^r\mid a^Tz=s\}\cap\Z^r.
$$
\end{lemma}

\begin{proof}
Recover $x_0\in C$ and successively maximize squared distance from the affine span of the points already found. If $r$ nonzero residuals, then $C$ is full dimensional. Otherwise $x_0,\ldots,x_k$, where $k<r$ and
$$
 \operatorname{aff}C=x_0+\operatorname{span}\{x_1-x_0,\ldots,x_k-x_0\}.
$$
Exact Gram--Schmidt gives an orthogonal matrix $Q$ whose first $k$ columns span $\operatorname{lin}(C-C)$.

Set $R=2\lceil R_{\rm out}\rceil+1$, $L=4R$, and $\varepsilon=(16rL)^{-r}$. Round $Q$ entrywise to denominator $N=\lceil4rR/\varepsilon\rceil+1$. The resulting rational matrix $\widetilde Q$ satisfies
$$
 \|\widetilde Q-Q\|_2<\frac{\varepsilon}{4R}<\frac12.
$$
Thus $\widetilde Q$ is nonsingular. With
$$
 B=\widetilde Q^{-T}\operatorname{diag}(LI_k,\varepsilon I_{r-k}),
$$
the last $r-k$ coordinates of $Q^T(x-x_0)$ vanish for $x\in C$, and the chosen bounds give
$$
 \|B^{-1}(x-x_0)\|_2^2
 \le
 \left(\frac{(1+\|\widetilde Q-Q\|_2)R}{L}\right)^2
 +\left(\frac{\|\widetilde Q-Q\|_2R}{\varepsilon}\right)^2
 <1.
$$
Hence $C\subseteq x_0+BB_2^r$. If $\kappa_r=\operatorname{vol}(B_2^r)$, then $k\le r-1$, every singular value of $\widetilde Q$ exceeds $1/2$, and $\kappa_r\ge(2/\sqrt r)^r$. Therefore
$$
 |\det B|
 \le2^rL^{r-1}\varepsilon
 <\frac{\kappa_r}{8^r}.
$$
The ellipsoid $\{y\mid4\|B^Ty\|_2\le1\}$ has volume greater than $2^r$. Minkowski's first theorem~\cite[Theorem~B.10]{DadushPeikertVempala2011} and Theorem~\ref{thmgeneralnormsvp} therefore give a shortest vector $B^Ta$ of $B^T\Z^r$ with $4\|B^Ta\|_2\le1$. The coefficient vector $a\in\Z^r$ is primitive.

Let $\ell=\min_{x\in C}a^Tx$ and $u=\max_{x\in C}a^Tx$. The ellipsoid containment gives $u-\ell\le2\|B^Ta\|_2\le1/2$. Thus $[\ell,u]$ contains no integer, or it contains the single integer $s=\lceil\ell\rceil$. These alternatives give the required certificate. The cost is $r^{O(r)}A$.
\end{proof}

The next lemma gives a rational Lenstra type normalization~\cite{Lenstra1983,DadushPeikertVempala2011}.

\begin{lemma}
\label{lemrationalnormalization}
Let $C\subseteq\R^r$ be compact, convex, and full dimensional, with exact point recovery, quadratic optimization, and
$$
 C\subseteq a_0+R_{\rm out}B_2^r.
$$
In $O(r)$ optimization calls and $r^{O(1)}$ further operations, one can compute an affine map $S(x)=Gx+g$ with nonsingular $G\in\Q^{r\times r}$ and a represented point $c$ such that
$$
 c+\frac{1}{4r^r}B_2^r
 \subseteq S(C)
 \subseteq c+3\sqrt r\,B_2^r.
$$
\end{lemma}

\begin{proof}
Choose $x_0\in C$. For $i=1,\ldots,r$, let $x_i\in C$ maximize squared distance from $\operatorname{aff}\{x_0,\ldots,x_{i-1}\}$. Let $h_i>0$ and $e_i$ be the length and unit direction of the residual, and define
$$
 T(x)_i=\frac{e_i^T(x-x_0)}{h_i}.
$$
Then $T(C)\subseteq[-1,1]^r$. If $v_i=T(x_i)$ and $V=(v_1,\ldots,v_r)$, then $V$ is upper triangular with diagonal one and $\|V\|_2\le r$, so $\sigma_{\min}(V)\ge r^{-(r-1)}$. The inball of the standard simplex gives, for $c=V\mathbf1/(r+\sqrt r)$,
$$
 c+\frac{1}{2r^r}B_2^r
 \subseteq T(C)
 \subseteq c+2\sqrt r\,B_2^r.
$$

Let $A$ be the linear part of $T$ and $\varepsilon=1/(8r^r)$. Round $A$ entrywise to denominator $N=\lceil rR_{\rm out}/\varepsilon\rceil+1$ and call the resulting rational matrix $G$. Set $S(x)=G(x-a_0)+T(a_0)$. Then
$$
 \sup_{x\in C}\|S(x)-T(x)\|_2
 \le\|G-A\|_2R_{\rm out}<\varepsilon.
$$
The support functions of $S(C)$ and $T(C)$ differ by less than $\varepsilon$. Hence
$$
 c+\left(\frac{1}{2r^r}-\varepsilon\right)B_2^r
 \subseteq S(C)
 \subseteq c+(2\sqrt r+\varepsilon)B_2^r,
$$
which gives the stated bounds. The inner bound implies that $G$ is nonsingular.
\end{proof}

The next lemma uses the width to shortest vector reduction from the proof of Theorem~4.7 in~\cite{DadushPeikertVempala2011}.

\begin{lemma}
\label{lemminimumwidth}
Let $C\subseteq\R^r$ be compact, convex, and full dimensional. Suppose that point recovery and exact linear and quadratic optimization over $C$ cost $A$, and that represented outer radius data are given. A primitive vector attaining $w_{\Z}(C)$ and the endpoints of its width interval can be computed in $r^{O(r)}A$ time.
\end{lemma}

\begin{proof}
Let $S(x)=Gx+g$ be given by Lemma~\ref{lemrationalnormalization}, put $D=S(C)-S(C)$, and let
$$
 D^*=\{y\in\R^r\mid y^Tz\le1\text{ for every }z\in D\}.
$$
Lemma~\ref{lemrationalnormalization} gives $(2r^r)^{-1}B_2^r\subseteq D\subseteq6\sqrt r\,B_2^r$, so $D^*$ is well-centered with Euclidean bounds depending only on $r$. For $a\in\Z^r$ and $y=G^{-T}a$, the reduction in~\cite[Theorem~4.7]{DadushPeikertVempala2011} gives
$$
 \|y\|_{D^*}
 =\max_{x\in C}a^Tx-\min_{x\in C}a^Tx.
$$
Membership in $D^*$ uses two exact linear optimizations over $C$. Since $a\mapsto G^{-T}a$ maps $\Z^r$ bijectively onto $G^{-T}\Z^r$, Theorem~\ref{thmgeneralnormsvp} returns a shortest vector whose coefficient vector attains $w_{\Z}(C)$. It is primitive, and two further optimizations give the interval endpoints. The cost is $r^{O(r)}A$.
\end{proof}

\subsection{Proof of Theorem~\ref{thmconvexintegerfeas}}

\begin{proof}
Consider a recursive call on a compact convex set
$C\subseteq\R^r$ in coordinates where the lattice is $\Z^r$.
First test emptiness and reject if $C$ is empty. If $r=0$, return
the unique point of $C$.

For $r>0$, apply Lemma~\ref{lemthindimensionreduction}. If it
certifies infeasibility, reject. If it returns a primitive
$a\in\Z^r$ and $s\in\Z$, choose
$W\in\operatorname{GL}_r(\Z)$ with $a^TW=e_r^T$. The map $t\longmapsto W\begin{pmatrix}t\\s\end{pmatrix}$ is a bijection from $\Z^{r-1}$ onto $\{z\in\Z^r\mid a^Tz=s\}$. Recurse on
$$
  \left\{
    t\in\R^{r-1}\ \middle|\
    W\begin{pmatrix}t\\s\end{pmatrix}\in C
  \right\}.
$$
If the recursive call returns $t$, return $W\begin{pmatrix}t\\s\end{pmatrix};$ otherwise reject. Thus, in the remaining case, $C$ is full dimensional.

Let $a$ be a minimum width direction from Lemma~\ref{lemminimumwidth}, and let $[\ell,u]$ be its width interval with $w=u-\ell$. If $w\le \phi(r)+1$, set $C_0=C$. Then every lattice point of $C$ lies on a level $a^Tz=s$ with $s\in\mathbb Z\cap[\ell,u]$.

If $w>\phi(r)+1$, let $x^-$ minimize $a^Tx$ and define
$$
  \alpha=\frac{\phi(r)+1}{w},
  \qquad
  C_0=(1-\alpha)x^-+\alpha C.
$$
Then $C_0\subseteq C$ and
$
  w_{\mathbb Z}(C_0)
  =\alpha w_{\mathbb Z}(C)
  =\phi(r)+1
  >\phi(r).
$
Hence the flatness theorem gives $\operatorname{int}(C_0)\cap\mathbb Z^r\ne\emptyset.$

Let
$
  \ell=\min_{z\in C_0}a^Tz,~
  u=\max_{z\in C_0}a^Tz.
$
In either case, $\mathbb Z\cap[\ell,u]$ contains at most
$\phi(r)+2$ integers.
Since $a$ is primitive, choose
$W\in\operatorname{GL}_r(\mathbb Z)$ such that
$a^TW=e_r^T$. For every
$s\in\mathbb Z\cap[\ell,u]$, define
$$
  \psi_s(t)
  =
  W\begin{pmatrix}t\\s\end{pmatrix}
$$
and recurse on
$$
  D_s
  =
  \left\{
    t\in\mathbb R^{r-1}\ \middle|\
    \psi_s(t)\in C_0
  \right\}.
$$
Accept and return $\psi_s(t)$ when one recursive call returns $t\in D_s\cap\mathbb Z^{r-1}$; reject if all recursive calls reject.

In the narrow case the selected sections contain every lattice point of $C$. In the wide case at least one selected section of $C_0$ contains a lattice point, and every returned point belongs to $C_0\subseteq C$.

The required operations and outer radius bounds are preserved under the sections, integral affine coordinate changes, and homotheties above. If $F(r)$ is the maximum cost at rank $r$, then
$$
  F(r)
  \le
  r^{O(r)}A+r^{O(1)}F(r-1),
$$
because $\phi(r)=r^{O(1)}$. Hence $F(r)=r^{O(r)}A$. The algorithm is deterministic.
\end{proof}

\section{Dual Bottleneck Steiner Value}
\label{appdualebst}

For a fixed Steiner point budget $j$, consider the bottleneck value $\beta_j(P)$ from Section~\ref{subsechybrid}. The parameterized result cited there already supplies the $j^{O(j)}n^{O(1)}$ bound. Full-topology enumeration also gives an $(n+j)^{O(n)}$ algorithm and an implicit optimizer. It uses the splitting and beading framework of~\cite[Section~2]{BrazilRasThomas2010}. In a full topology, the budget becomes $\mathbf 1^Tq\le n+j-1$, while the objective is the bottleneck length.

Fix $n\ge3$, a full Steiner topology $H=(V,E)$, and an integer budget $j\ge0$.
Set
$$
 Q_j=n+j-1,
 \qquad
 \Delta_B=\operatorname{diam}_p(B(P)).
$$
Let $\rho_{H,j}$ be the value of the following semialgebraic program with integer variables. Its value is $+\infty$ when the program is infeasible.
\begin{equation}
\label{eqdualfixedtopology}
\begin{aligned}
 \rho_{H,j}=\min\quad & t \\
 \text{s.t.}\quad
 & q\in\Z_{\ge0}^{E},
 \qquad \mathbf 1^Tq\le Q_j,\\
 & x_s\in B(P) \qquad(s\in S(H)),\\
 & 0\le t\le \Delta_B,\\
 & \mathcal B_p(X_u(x)-X_v(x),tq_e)
 \qquad(e=uv\in E).
\end{aligned}
\end{equation}

\begin{lemma}
\label{lemdualfullequivalence}
For $n\ge3$ and $j\ge0$,
$$
 \beta_j(P)=\min_{H\in\mathcal F_n}\rho_{H,j}.
$$
\end{lemma}

\begin{proof}
Let $(H,q,x,t)$ be feasible. After zero count contractions and positive edge subdivision, every segment has length at most $t$, and the counting argument of Lemma~\ref{lembeadcount} gives at most
$$
 \mathbf 1^Tq-n+1\le j
$$
Steiner points before identifications. Thus $\beta_j(P)\le\min_H\rho_{H,j}$.

Conversely, take an optimal bottleneck tree with budget $j$ by the leaf deletion, identification, straightening, and splitting steps used in the proof of Lemma~\ref{lembeadcount}. None increases the bottleneck. Let $q_e$ count the path edges represented by a positive full topology edge and set $q_e=0$ on split edges. Then $\mathbf 1^Tq\le n+j-1$, and every full topology edge has length at most $\beta_j(P)q_e$. Since a terminal spanning tree has bottleneck at most $\Delta_B$, we have $\beta_j(P)\le\Delta_B$. Projection into $B(P)$ preserves all edge inequalities, so the program is feasible with $t=\beta_j(P)$. Hence $\min_H\rho_{H,j}\le\beta_j(P)$.
\end{proof}

\begin{lemma}
\label{lemdualfixedcomputation}
For fixed $H\in\mathcal F_n$, the value $\rho_{H,j}$ and a minimizing
witness, when finite, are computable in
$$
 n^{O(n)}\binom{j+3n-4}{2n-3}.
$$
\end{lemma}

\begin{proof}
Let $d=|E(H)|=2n-3$. The number of vectors $q\in\Z_{\ge0}^{E}$ satisfying $\mathbf 1^Tq\le Q_j=n+j-1$ is
$$
 \binom{Q_j+d}{d}=\binom{j+3n-4}{2n-3}.
$$
For each fixed $q$, the remaining problem is a compact semialgebraic minimization with $O(n)$ variables and constraints of constant degree, and Theorem~\ref{thmsemialgopt} solves it in $n^{O(n)}$. Enumerating all $q$ and retaining the least feasible value gives the stated bound and a minimizing witness.
\end{proof}

\begin{theorem}
\label{thmdualebstdirect}
For every $j\ge0$, the value $\beta_j(P)$ for the bottleneck
Steiner problem with at most $j$ Steiner points is computable
deterministically in
$$
 (n+j)^{O(n)}.
$$
For $n\ge3$, the full topology algorithm returns a full topology, branch coordinates, and segment counts whose expansion attains $\beta_j(P)$ with at most $j$ Steiner points.
\end{theorem}

\begin{proof}
For $n=1$, the value is 0. For $n=2$, equal subdivision of $p_1p_2$ into $j+1$ segments gives bottleneck $\norm{p_1-p_2}/(j+1)$. Every path between the two terminals with at most $j$ Steiner points has at most $j+1$ edges, so the triangle inequality gives the lower bound.

Assume $n\ge3$. Enumerate $H\in\mathcal F_n$, compute $\rho_{H,j}$ by Lemma~\ref{lemdualfixedcomputation}, and return the least finite value and a witness. Lemma~\ref{lemdualfullequivalence} proves correctness. Lemma~\ref{lemfullcount} gives $n^{\Theta(n)}$ topologies, and
$$
 n^{O(n)}\binom{j+3n-4}{2n-3}=(n+j)^{O(n)}.
$$
The returned witness expands to at most $\mathbf 1^Tq-n+1\le j$ Steiner points. Any further geometric identification and deletion of cycle edges can only reduce this number. The expansion has at most $n+j$ vertices.
\end{proof}

\bibliographystyle{cas-model2-names.bst}
\bibliography{references}

@article{BaeEtAl2011,
  author  = {Sang Won Bae and Sunghee Choi and Chunseok Lee and Shin-ichi Tanigawa},
  title   = {Exact Algorithms for the Bottleneck Steiner Tree Problem},
  journal = {Algorithmica},
  volume  = {61},
  number  = {4},
  pages   = {924--948},
  year    = {2011},
  doi     = {10.1007/s00453-011-9553-y}
}

@inproceedings{BandyapadhyayEtAl2024,
  author    = {Sayan Bandyapadhyay and William Lochet and Daniel Lokshtanov and Saket Saurabh and Jie Xue},
  title     = {Euclidean Bottleneck Steiner Tree is Fixed-Parameter Tractable},
  booktitle = {Proceedings of the 2024 Annual {ACM-SIAM} Symposium on Discrete Algorithms},
  pages     = {699--711},
  year      = {2024},
  doi       = {10.1137/1.9781611977912.27}
}

@book{BasuPollackRoy2006,
  author    = {Saugata Basu and Richard Pollack and Marie-Fran{\c c}oise Roy},
  title     = {Algorithms in Real Algebraic Geometry},
  series    = {Algorithms and Computation in Mathematics},
  volume    = {10},
  edition   = {2},
  publisher = {Springer},
  address   = {Berlin, Heidelberg},
  year      = {2006},
  doi       = {10.1007/3-540-33099-2}
}

@article{BrazilRasThomas2010,
  author  = {Marcus Brazil and Charl J. Ras and Doreen A. Thomas},
  title   = {Approximating Minimum Steiner Point Trees in Minkowski Planes},
  journal = {Networks},
  volume  = {56},
  number  = {4},
  pages   = {244--254},
  year    = {2010},
  doi     = {10.1002/net.20376}
}

@article{BrazilEtAl2011Generalised,
  author  = {Marcus N. Brazil and Charl J. Ras and Konrad J. Swanepoel and Doreen A. Thomas},
  title   = {Generalised {$k$}-Steiner Tree Problems in Normed Planes},
  journal = {Algorithmica},
  volume  = {71},
  number  = {1},
  pages   = {66--86},
  year    = {2015},
  doi     = {10.1007/s00453-013-9780-5}
}

@article{CalinescuWang2023,
  author  = {Gruia C{\u a}linescu and Xiaolang Wang},
  title   = {Combination Algorithms for Steiner Tree Variants},
  journal = {Algorithmica},
  volume  = {85},
  pages   = {153--169},
  year    = {2023},
  doi     = {10.1007/s00453-022-01009-8}
}

@article{ChenEtAl2001TCS,
  author  = {Donghui Chen and Ding-Zhu Du and Xiao-Dong Hu and Guo-Hui Lin and Lusheng Wang and Guoliang Xue},
  title   = {Approximations for Steiner Trees with Minimum Number of Steiner Points},
  journal = {Theoretical Computer Science},
  volume  = {262},
  number  = {1--2},
  pages   = {83--99},
  year    = {2001},
  doi     = {10.1016/S0304-3975(00)00182-1}
}

@article{ChengEtAl2008,
  author  = {Xiuzhen Cheng and Ding-Zhu Du and Lusheng Wang and Baogang Xu},
  title   = {Relay Sensor Placement in Wireless Sensor Networks},
  journal = {Wireless Networks},
  volume  = {14},
  number  = {3},
  pages   = {347--355},
  year    = {2008},
  doi     = {10.1007/s11276-006-0724-8}
}

@article{CohenNutov2018,
  author  = {Nachshon Cohen and Zeev Nutov},
  title   = {Approximating Steiner Trees and Forests with Minimum Number of Steiner Points},
  journal = {Journal of Computer and System Sciences},
  volume  = {98},
  pages   = {53--64},
  year    = {2018},
  doi     = {10.1016/j.jcss.2018.08.001}
}

@article{LinXue1999,
  author  = {Guang-Hong Lin and Guoliang Xue},
  title   = {Steiner Tree Problem with Minimum Number of Steiner Points and Bounded Edge-Length},
  journal = {Information Processing Letters},
  volume  = {69},
  number  = {2},
  pages   = {53--57},
  year    = {1999},
  doi     = {10.1016/S0020-0190(98)00201-4}
}

@article{MandoiuZelikovsky2000,
  author  = {Ion I. M{\u a}ndoiu and Alexander Z. Zelikovsky},
  title   = {A Note on the {MST} Heuristic for Bounded Edge-Length Steiner Trees with Minimum Number of Steiner Points},
  journal = {Information Processing Letters},
  volume  = {75},
  number  = {4},
  pages   = {165--167},
  year    = {2000},
  doi     = {10.1016/S0020-0190(00)00095-8}
}

@article{NutovYaroshevitch2009,
  author  = {Zeev Nutov and Ariel Yaroshevitch},
  title   = {Wireless Network Design via 3-Decompositions},
  journal = {Information Processing Letters},
  volume  = {109},
  number  = {19},
  pages   = {1136--1140},
  year    = {2009},
  doi     = {10.1016/j.ipl.2009.07.013}
}

@article{ShinChoi2023,
  author  = {Donghoon Shin and Sunghee Choi},
  title   = {An Efficient 3-Approximation Algorithm for the Steiner Tree Problem with the Minimum Number of Steiner Points and Bounded Edge Length},
  journal = {PLOS ONE},
  volume  = {18},
  number  = {11},
  pages   = {e0294353},
  year    = {2023},
  doi     = {10.1371/journal.pone.0294353}
}

@misc{DadushVempala2012,
  author        = {Daniel Dadush and Santosh S. Vempala},
  title         = {Near-Optimal Deterministic Algorithms for Volume Computation and Lattice Problems via M-Ellipsoids},
  year          = {2012},
  eprint        = {1201.5972},
  archivePrefix = {arXiv},
  primaryClass  = {cs.DS}
}

@article{Lenstra1983,
  author  = {Lenstra, Jr., Hendrik W.},
  title   = {Integer Programming with a Fixed Number of Variables},
  journal = {Mathematics of Operations Research},
  volume  = {8},
  number  = {4},
  pages   = {538--548},
  year    = {1983},
  doi     = {10.1287/moor.8.4.538}
}

@misc{DadushPeikertVempala2011,
  author        = {Dadush, Daniel and Peikert, Chris and Vempala, Santosh},
  title         = {Enumerative Lattice Algorithms in Any Norm via {M}-Ellipsoid Coverings},
  year          = {2011},
  note          = {Proceedings of the 52nd Annual IEEE Symposium on Foundations of Computer Science},
  eprint        = {1011.5666},
  archivePrefix = {arXiv},
  primaryClass  = {cs.DS}
}

@article{BanaszczykEtAl1999,
  author  = {Banaszczyk, Wojciech and Litvak, Alexander E. and Pajor, Alain and Szarek, Stanis{\l}aw J.},
  title   = {The Flatness Theorem for Nonsymmetric Convex Bodies via the Local Theory of Banach Spaces},
  journal = {Mathematics of Operations Research},
  volume  = {24},
  number  = {3},
  pages   = {728--750},
  year    = {1999},
  doi     = {10.1287/moor.24.3.728}
}

@article{Rudelson2000,
  author  = {Rudelson, Mark},
  title   = {Distances between Non-Symmetric Convex Bodies and the {$MM^*$}-Estimate},
  journal = {Positivity},
  volume  = {4},
  number  = {2},
  pages   = {161--178},
  year    = {2000},
  doi     = {10.1023/A:1009842406728}
}

\end{document}